\documentclass[a4paper,11pt]{article}
\usepackage{amsthm}
\usepackage[margin=2.8cm]{geometry}
\usepackage[colorlinks=true, urlcolor=blue]{hyperref}

\usepackage{todonotes}
\usepackage{amsfonts}
\usepackage{amsmath}
\usepackage{amsfonts,latexsym,amssymb,blindtext,titlesec,algorithm2e,mathrsfs,algpseudocode}
\usepackage{authblk}
\usepackage{xcolor}
\usepackage{tcolorbox}
\tcbset{colback=gray!10, colframe=gray!40, boxrule=0pt, arc=2mm, left=4mm, right=4mm, top=2mm, bottom=2mm,
 width=0.75\textwidth, center title}

\usepackage{graphicx}

\newtheorem{definition}[]{Definition}
\newtheorem{lemma}[]{Lemma}
\newtheorem{corollary}[]{Corollary}
\newtheorem{claim}[]{Claim}

 \RestyleAlgo{ruled}

\DeclareMathOperator{\poly}{poly}

\newcommand{\LP}{\textnormal{\textsf{LP}}\xspace}
\newcommand{\LPstar}{\textnormal{\textsf{LP\textsuperscript{*}}}\xspace}
\newcommand{\NLP}{\textnormal{\textsf{NLP}}\xspace}
\newcommand{\PLD}{\textnormal{\textsf{PLD}}\xspace}
\newcommand{\LD}{\textnormal{\textsf{LD}}\xspace}
\newcommand{\NLD}{\textnormal{\textsf{NLD}}\xspace}
\newcommand{\NPLD}{\textnormal{\textsf{NPLD}}\xspace}

\newcommand{\local}{\textnormal{\textsf{local}}}
\newcommand{\Plocal}{\textnormal{\textsf{P-local}}}

\newcommand{\PiOneLocal}{\Pi_1^{\local}\xspace}
\newcommand{\PiOnePLocal}{\Pi_1^{\Plocal}\xspace}
\newcommand{\PiOneLP}{\Pi_1^{\LP}\xspace}

\newcommand{\SigmaOneLocal}{\Sigma_1^{\local}\xspace}
\newcommand{\SigmaOnePLocal}{\Sigma_1^{\Plocal}\xspace}
\newcommand{\SigmaOneLP}{\Sigma_1^{\LP}\xspace}

\newcommand{\PH}{\textsf{PH}\xspace}
\newcommand{\genericalgo}{\textnormal{\textbf{A}}\xspace}
\newcommand{\genericlang}{{\textnormal{\texttt{L}}}\xspace}

\newcommand{\classP}{\textnormal{\textsf{P}}\xspace}
\newcommand{\classNP}{\textnormal{\textsf{NP}}\xspace}

\newcommand{\AGTG}{\textnormal{\texttt{AGTG}}\xspace}
\newcommand{\ALTG}{\textnormal{\texttt{ALTG}}\xspace}
\newcommand{\node}{\textnormal{\texttt{node}}\xspace}
\newcommand{\pathlang}{\textnormal{\texttt{path}}\xspace}
\newcommand{\Lsp}{\textnormal{\texttt{L}}_{\texttt{SP}}\xspace}

 \title{Polynomial Time Local Decision Revisited}
 \author[1]{Laurent Feuilloley%
 \thanks{Supported by ANR project PREDICTIONS (ANR-23-CE48-0010).}}
 \author[2]{Soumyadeep Paul}
 \author[3]{Ami Paz%
 \thanks{Supported by ANR project DIDYA (ANR-25-CE48-0897).}}
 \affil[1]{CNRS, INSA Lyon, UCBL, LIRIS, UMR5205, F-69622 Villeurbanne, France }
 \affil[2]{Tata Institute of Fundamental Research, Mumbai}
 \affil[3]{LISN --- CNRS \& Paris-Saclay University}
\date{}

\begin{document} 

\maketitle
\begin{abstract}
We consider three classification systems for distributed decision tasks: With unbounded computation and certificates, defined by Balliu, D'Angelo, Fraigniaud, and Olivetti [JCSS'18], and with (two flavors of) polynomially bounded local computation and certificates, defined in recent works by Aldema Tshuva and Oshman [OPODIS'23], and by Reiter [PODC'24]. The latter two differ in the way they evaluate the polynomial bounds: the former considers polynomials with respect to the size of the graph, while the latter refers to being polynomial in the size of each node's local neighborhood.

We start by revisiting decision without certificates. For this scenario, we show that the latter two definitions coincide: roughly, a node cannot know the graph size, and thus can only use a running time dependent on its neighborhood.

We then consider decision with certificates. With existential certificates ($\Sigma_1$-type classes), a larger running time defines strictly larger classes of languages: when it grows from being polynomial in each node's view, through polynomial in the graph's size, and to unbounded, the derived classes strictly contain each other. With universal certificates ($\Pi_1$-type classes), on the other hand, we prove a surprising incomparability result: having running time bounded by the graph size sometimes allows us to decide languages undecidable even with unbounded certificates.

We complement these results with other containment and separation results, which together portray a surprisingly complex lattice of strict containment relations between the classes at the base of the three classification systems.
\end{abstract}

 \thispagestyle{empty}
\newpage
\tableofcontents

 \thispagestyle{empty}

\clearpage
\setcounter{page}{1}

\section{Introduction}

Decision tasks constitute the core of computational complexity theory, and their role in theoretical computer science is crucial.
In his 2010 talk~\cite{Fraigniaud10}, Fraigniaud suggested distributed decision tasks as a base for a complexity theory for distributed computing. 
Following this, Fraigniaud, Korman, and Peleg~\cite{FraigniaudKP13} started the systematic study of distributed decision tasks, a project that continues until today.
They defined several classes of graph languages, where the basic one is \LD, consisting of all languages that can be decided by the following procedure.
The input is a labeled graph, representing a communication network; the nodes of the graph communicate for a constant time, and then a Turing machine at each node performs a time-unbounded local computation and decides if to accept or reject. 
An input graph is accepted if all the Turing machines at its nodes accept.
A simple example is the language of properly $k$-colored graphs: each node is labeled by a color in $\{1,\ldots,k\}$, and in the decision algorithm accepts if its color is different from those of its neighbors. 
It is easy to see that if the graph is properly $k$-colored then all nodes accept, while otherwise at least two of them reject.

The same work also defines the class \NLD, where each node is equipped with a \emph{certificate} of unbounded size, and a configuration is in the language if there exists a certificate assignment that makes all the nodes accept.
This latter notion is very similar to the notion of \emph{proof labeling schemes}, introduced by Korman, Kutten, and Peleg~\cite{KormanKP10} a few years beforehand, and coincides with it for many languages.
An example of a certification scheme (valid both for \NLD and as a Proof Labeling Scheme) is for $k$-colorability, the claim that the (initially unlabeled) graph has a proper $k$-coloring.
To certify this, each node gets a color in $\{1,\ldots,k\}$ as a certificate, and accepts iff its certificate is different from those of its neighbors.

Later on, Balliu, D'Angelo, Fraigniaud, and Olivetti~\cite{BalliuDFO18} introduced a hierarchy for distributed decision inspired by the polynomial hierarchy in classic complexity theory. 
This hierarchy captured the classes already known and allowed to define new classes of distributed decision tasks.
They studied the relation between these new classes, and in particular they showed that their hierarchy collapses to the second level: $\Pi_2^{\local}$ contains all decidable distributed languages.

An interesting artifact of their definition is that any language that is stable under the so-called lift operation can be certified using $O(n^2)$-bit certificates.
This is a very large class of language,
and contains, e.g., the language of all non-$3$-colorable graphs.
To certify such a property, each node is simply given a map of all the graph, and verifies locally that it is consistent with its view and that the given graph has the property (e.g., is not $3$-colorable).
This is allowed since the local computation at each node is unbounded and not accounted for, as is common in distributed graph algorithms.

Recently, Aldema Tshuva and Oshman~\cite{TshuvaO23} considered the case where local computation is bounded, limited to polynomial in the size of the graph. 
Shortly after, Reiter~\cite{Reiter24} suggested a similar model, but where the local computation is polynomial in the view of the node, that is, in the size of the neighborhood rather than the full graph.
Both these works defined hierarchies that resemble the centralized polynomial hierarchy, and studied relations between the classes in these hierarchies.
Aldema Tshuva and Oshman showed that their hierarchy coincides with the centralized hierarchy (restricted to graphs) starting from the second level. 
Reiter proved that his hierarchy coincides with the centralized one for graphs of a single node.
A common theme of these works is that they essentially resolve the relations from level 2 of the hierarchies and above, while leaving intriguing open questions regarding the relations in the lower levels of the hierarchies, and between them.

\subsection{Our Results}
\label{sec:our-results-intro}

The central goal of our work is to understand the relations between the distributed complexity classes defined in the aforementioned works. 
While much of this comparison has been carried out for the higher levels of the hierarchy by Aldema Tshuva and Oshman~\cite{TshuvaO23} and  Reiter~\cite{Reiter24}, several gaps remain in the lower hierarchy levels. 
Our main results are depicted in Figure~\ref{fig:diagram}.

\begin{figure}[tb]
    \centering
    \begin{tikzpicture}[
  every node/.style={align=center},
  solid/.style={->, thick}
]

\node (lpstr) at (3,0.5) {$\LPstar =\PLD$\\{\footnotesize Lem.~\ref{lem:lpeqpld}}};
\node (lp) at (5,0) {$\LP$};
\node (pionelp) at (-0.5,2) {$\Pi_1^{\LP}$};
\node (pioneplocinter) at (4.4,3.0) {$\PiOnePLocal\cap \SigmaOnePLocal$};
\node (nlp) at (7,2) {$\SigmaOneLP$};
\node (ld) at (2.2,3) {$\LD$};
\node (pioneploc) at (-0.5,5) {$\PiOnePLocal$};
\node (pioneplocn) at (-2.2,7) {$\Pi_1^{\text{\Plocal}[n]}$};
\node (npld) at (7,5) {$\SigmaOnePLocal$};
\node (pioneloc) at (0.5,7.7) {$\PiOneLocal$};
\node (nld) at (6,7.7) {$\SigmaOneLocal$};

\draw[solid] (lp) -- (lpstr);
\draw[solid] (lpstr) -- (pioneplocinter) node [midway, fill=white] {\footnotesize ~~Cor.~\ref{cor:PLD in intersection}\\\footnotesize ~~\cite{TshuvaO23}};
\draw[solid] (lp) -- (nlp) node [midway, fill=white] {\footnotesize \cite{Reiter24}};
\draw[solid] (pionelp) -- (ld)  node [midway, fill=white] {\footnotesize Lem.~\ref{lem:pi<ld}};
\draw[solid] (pionelp) -- (pioneploc) node [midway, fill=white] {\footnotesize Lem.~\ref{lem:pilp<pioneploc}} ;
\draw[solid] (pioneplocinter) -- (pioneploc);
\draw[solid] (ld) -- (pioneloc) node [midway, fill=white] {\footnotesize \cite{BalliuDFO18}};
\draw[solid] (nlp) -- (npld) node [midway, fill=white] {\footnotesize \cite{Reiter24}};
\draw[solid] (pioneplocinter) -- (npld);
\draw[solid] (npld) -- (nld) node [midway, fill=white] {\footnotesize Cor.~\ref{cor:PHvsLD}};
\draw[solid] (lpstr) -- (ld) node [midway, fill=white] {\footnotesize Lem.~\ref{lem:PLDsneLD}};
\draw[dashed, thick] (pioneploc) -- (pioneloc) node [midway, fill=white] {\footnotesize Lem.~\ref{lem:pip<pi}};

\draw[solid] (pioneploc) -- (pioneplocn) node [midway, fill=white] {\footnotesize Lem.~\ref{lem:piloc<pilocn}};

\draw[dashed,thick] (pioneloc) -- (pioneplocn) node [midway, above] {\footnotesize Cor.~\ref{cor:pioneplocn<pioneloc}};

\draw[solid] (ld) -- (nld) node [midway, fill=white] {\footnotesize \cite{BalliuDFO18}};

\end{tikzpicture}
    \caption{We consider the classes with universal quantifiers ($\Pi_1$-type, left), deterministic computation (center) and
existential quantifiers ($\Sigma_1$-type, right). 
Solid arrows represent strict containment, except for $\LP\subseteq\LPstar$ where the strictness remains open. 
The dotted line refers to the fact that the classes are incomparable. Only the known relations that are relevant to the paper are represented on the figure.}
    \label{fig:diagram}
\end{figure}

Intuitively, one would expect the following relation between classes of decision tasks, both regarding computation time and certificate sizes. 
\begin{center}
\begin{tcolorbox}
\[
\begin{aligned}
  \begin{matrix}
    \text{polynomial in} \\[2pt]
    \text{local view}
  \end{matrix}
  &\;\subsetneqq\;&
  \begin{matrix}
    \text{polynomial in} \\[2pt]
    \text{graph size}
  \end{matrix}
  &\;\subsetneqq\;&
  \text{unbounded}
\end{aligned}
\]
\end{tcolorbox}
\end{center}
Surprisingly, this turns out not to always be the case. 
While this is true for the existential quantifier (classes of type $\Sigma_1$ defined with ``$\exists$ certificate $c$''),
this is not the case without certificates, or with the universal quantifier (class of type $\Pi_1$ defined with ``$\forall$ certificate~$c$'').

We restrict our attention to constant verification radius, do not consider randomization, and focus on ID-oblivious certification, that is, the certificate assignment is independent of the IDs. 

\subsubsection{Without certificates: Equality of LP* and PLD} 

\label{sec:lp-pld-intro}
The first, and perhaps most fundamental comparison we present is between certificate-free classes: the class \LPstar, which is closely related to the class \LP
defined by Reiter~\cite{Reiter24},
and the class \PLD defined by Aldema Tshuva and Oshman~\cite{TshuvaO23}.
In both \LP and \LPstar, a node must decide by performing computation in time polynomial in the size of its local view, whereas in \PLD, decisions must be taken in time polynomial in the graph size.
The subtle difference between \LP and \LPstar is that
an \LPstar algorithm only needs to be correct for globally unique IDs, while an \LP algorithm must be correct for any locally unique ID assignment.
Therefore, an \LP algorithm is also an \LPstar algorithm, i.e., $\LP \subseteq \LPstar$.
All these three classes are contained in \LD, where the local computation is unbounded.

The containment  $\LPstar \subseteq \PLD$ is almost trivial, and one is tempted to speculate that proving $\LPstar \subsetneqq \PLD$ is just as simple. 
After all, in a low-degree graph \LPstar allows much less local computation time than \PLD.
However, in Section~\ref{sec:lp-pld}, we prove this speculation to be wrong.
\begin{center}
\begin{tcolorbox}
\[
\begin{aligned}
    \LP
    &\;\subseteq\;&
    \LPstar
  &\;=\;&
    \PLD
  &\;\subsetneqq\;&
  \LD
\end{aligned}
\]
\end{tcolorbox}
\end{center}

Roughly speaking, we show that a \PLD algorithm cannot fully utilize its supposed global polynomial-time bound, since it lacks knowledge of the total graph size.

Continuing this section, we show that this simple idea is instrumental in proving separations between distributed decision classes.
One result asserts that  $\PLD \subsetneqq \classP \cap \LD$, which also implies $\PLD \subsetneqq \LD$.
That is, polynomial time local decision is strictly weaker than centralized polynomial-time decision (\classP) and local decision with unbounded local time (\LD).
This was proved in~\cite{TshuvaO23} using a carefully crafted language called \small{ITER-BOUND}, and applying a computability argument;
here, we give a much simpler proof, based on the fact that the local polynomial time cannot be fully utilized.

\subsubsection{Existential Certificates}

In the case of existential certificates, the classes $\SigmaOneLP(\NLP),\SigmaOnePLocal(\NPLD)$ and $\SigmaOneLocal(\NLD)$ behave ``as expected''.
\begin{center}
\begin{tcolorbox}
\[
\begin{aligned}
    \Sigma_1^{\text{\textsf{LP}}} (\NLP)
  &\;\subsetneqq\;&
    \Sigma_1^{\text{\textsf{P-local}}}\;(\NPLD)
  &\;\subsetneqq\;&
  \Sigma_1^{\text{\textsf{local}}} (\NLD)
\end{aligned}
\]
\end{tcolorbox}\end{center}

The containments in both cases follow easily from the definitions since an existential certificate from a stricter space is still an existential certificate when we are allowed larger certificates. 
Note that this simple argument does not work when we consider the corresponding $\Pi$ classes below: in the case of $\Pi$ classes we need to ensure that the algorithm accepts for all certificates in a larger set of possible certificates.
Intuitively speaking, for $\Sigma$ we can ``use the same certificate'', while for $\Pi$ we ``need to accept on more certificates''.

The inequality $\NLP \neq \NPLD$ was proved in \cite{Reiter24} using the language {\footnotesize\texttt{NOT-ALL-SELECTED}}, which is the language of graphs where the nodes are given labels $0$ or $1$ and at least one of the nodes is given the label $0$.
The inequality $\NPLD \neq \NLD$ is a simple consequence of Corollary~\ref{cor:PHvsLD} we prove at the end of Section~\ref{sec:perlim}.
There, we define the language $\node_{\Lsp}$ which is in $\LP$ but not in the the centralized polynomial hierarchy $\PH$, while $\NPLD$ is clearly a subset of $\PH$.

\subsubsection{Universal Certificates: Certificate-Bounded Separations}
\label{sec:certificates-intro}

Up until here, we showed that for local computation, time that is polynomial in the graph size is not more powerful than having it polynomial only in the local view. 
For existential certificates, we showed that certificates with size polynomial in the graph size are slightly more powerful, and lies between certificates with size bounded by the local view size, and unbounded certificates.
In Section~\ref{sec:poly-certificates} we move to consider universal certificates.

In Section~\ref{subsec:bounded-leak} we present a result of the opposite flavor:
In some cases, 
bounding the certificates by the graph size lets us solve problems which would not be solvable with unbounded certificates. 
Concretely, $\PiOnePLocal
  \nsubseteq
  \PiOneLocal$.
To prove this separation,
we use the language \AGTG of $n$-node paths with each node $i$ labeled $x_i$, such that $x_i>n$ for all $i$.
There is a simple 
$\PiOnePLocal$ algorithm for it, when bounding the certificate sizes by the polynomial $Q(n)=n$: node $i$ checks that $x_i$ is larger than the size of its certificate.
As another application of the \AGTG language, we prove that  $\PLD \subsetneqq \PiOnePLocal \cap \NPLD$.
This was claimed in~\cite{TshuvaO23}, and here we provide a simple, locality-based proof of it.

To prove the aforementioned separation,
we use the fact that the nodes can deduce information on the value of $n$.
In their work~\cite{TshuvaO23}, Aldema Tshuva and Oshman defined the class $\Pi_1^{\text{\Plocal}[n]}$, which is the same as the class $\PiOnePLocal$, but with the nodes knowing $n$.
Since \AGTG used to separate 
$\PiOnePLocal$ and $\PiOneLocal$ is using information on $n$ learned by the nodes in a $\PiOnePLocal$ algorithm, one may suspect that the nodes can learn $n$, yielding
$\Pi_1^{\text{\textsf{P-local}}[n]}=\PiOnePLocal$.
In Section~\ref{subsec:knowing n} we disprove this suspicion, and show
$\PiOnePLocal\subsetneqq\Pi_1^{\text{\Plocal}[n]}$.
This, on the other hand, begs the question whether knowing $n$ makes the class $\Pi_1^{\text{\Plocal}[n]}$  as strong as $\PiOneLocal$, which we also answer in the negative: we show that 
$\Pi_1^{\text{\Plocal}[n]}$ and $\PiOneLocal$ are incomparable.

In Section~\ref{subsec:LPcertificates} we consider the ``weaker'' class $\PiOneLP$, where the certificate sizes are bounded by the local view size.
We complement the aforementioned inequality $\PiOnePLocal
  \nsubseteq
\PiOneLocal$
by showing that 
$\PiOneLP$
is strictly contained in both classes.
The containment claims were claimed by Reiter~\cite{Reiter24} without a proof, while the separation was left open in his work.
For both cases, we complement Reiter's containment claims with formal proofs, and prove the corresponding  separation results.

\begin{center}
\begin{tcolorbox}
\[
\begin{aligned}
    \PiOneLP
  &\;\subsetneqq\;&
    \PiOnePLocal
  &\;\nsubseteq\;&
  \PiOneLocal
\end{aligned}
\]
\end{tcolorbox}
\end{center}

\begin{center}
\begin{tcolorbox}
\[
\begin{aligned}
    \PiOneLP
  &\hspace{3ex}&
   \subsetneqq
  &\hspace{3ex}&
  \PiOneLocal
\end{aligned}
\]
\end{tcolorbox}
\end{center}

In fact, in order to prove $\PiOneLP \subsetneqq \PiOneLocal$, we prove a more refined result:
$\PiOneLP \subsetneqq \LD$. 
To this end, we show that a $\PiOneLP$ algorithm can be simulated by an \LD algorithm of double the locality radius.
It is unclear whether a simulation with the same locality is possible.
A similar argument is used to simulate a $\PiOneLP$ algorithm by a 
$\PiOnePLocal$ algorithm.

\section{Model and Definitions}
\label{sec:perlim}

Let us define the model formally. We follow closely the definitions and notations of~\cite{TshuvaO23}.

\subsection{General Definitions}
We consider simple graphs on $n$ nodes.
In some cases the nodes will have input labels, for example a color for each node. 
A \emph{graph configuration} $(G,x)$ is a graph $G$, equipped with an assignment of inputs, $x:V \to S$ (where $S$ is some input set).
We assume that all the inputs are bounded by some polynomial in the size of the graph.
We use $N_G^t(v)$ to denote the distance-$t$ neighborhood of $v$, and $\#N_G^t(v)$ the number of nodes in it.
We also use $|N^{t}_{G,x,id}(v)|$ for the number of bits in $v$'s distance-$t$ neighborhood, where we take into account the inputs, IDs and certificates given to the nodes (see definitions below).
 
\subsubsection{Identifiers}
\label{sec:ids}
The distributed certification literature contains several models for identifiers.
Unless stated otherwise, 
we use the classic version where  
every node is given a unique integer encoded on $O(\log n)$ bits. 
For most classes, that is \LD and \PLD, and the ones with superscripts \local{} and \Plocal{},
the certificates should not depend on the identifiers, in the sense that the identifiers are chosen adversarially after the certificates are chosen. 
See~\cite{Feuilloley21, FeuilloleyF16} for discussions of other models, and~\cite{FeuilloleyFH21} for  another convention.

The identifiers in the \LP{} class and the ones with \LP{} superscript behave differently. 
As computation time in these classes depends only on the local view, the identifier's uniqueness need not be universal: when considering a $t$-local algorithm, the identifiers in every distance-$t$ ball around a node must be distinct, but not in the whole graph.

\subsubsection{Languages, Neighborhoods, Local Decision Algorithms}
In local decision, it is common to use the vocabulary from language theory, and we define a \emph{distributed language} as a set of graph configurations. 
In the following, we might simply refer to a \emph{language}. 
Note that the identifiers do not appear in the definition of a distributed language. An example of a language is the set of all the graphs for which the inputs encode a proper coloring.

Let $N^t_{G,x,id}(v)$ denote the $t$-neighborhood of $v$ in the identified configuration $(G, x, id)$,
including the identifiers and the inputs of the nodes in the $t$-neighborhood. 
Let $N^t_{G,x}(v)$ be the same but without the identifiers.

Next, we formally define local distributed algorithms. 
In a distributed algorithm, each
node observes the neighborhood around itself and then decides whether to accept or reject.

\begin{definition}[Local decision algorithms]
A $t$-local decision algorithm $\genericalgo$ is a computable
mapping from identified neighborhoods of size $t$ to Boolean output values, \emph{accept} or \emph{reject}.
When $t$ is constant, we refer to such an algorithm as an \emph{\LD-algorithm}.

If $\genericalgo(N^{t}_{G,x,id}(v))$ accepts 
at all nodes $v \in  V (G)$, then we say that $\genericalgo$ \emph{accepts} $(G, x, id)$.
We say that $\genericalgo$ \emph{decides} the distributed language $\genericlang$ if for every graph configuration $(G, x)$ and for
every identifier assignment $id$,
$(G, x) \in  \genericlang \Leftrightarrow \genericalgo(G, x, id) = $ \emph{accept}.
Given a $t$-local decision algorithm, we refer to~$t$ as the algorithm’s \emph{locality radius}.
\end{definition}

\subsubsection{Deterministic Classes}

We first define the deterministic class without local computational constraints.

\begin{definition}[The class \LD]
A distributed language
$\genericlang$ is in the class \LD if it can be decided by a $t$-local decision
algorithm $\genericalgo$ for some constant $t$. 
\end{definition}

\begin{definition}[The class \PLD]
A distributed language
$\genericlang$ is in the class \PLD if there exist a polynomial $P$ such that for all $n>0$, $\genericlang$ can be decided on any $n$-node graph by an \LD-algorithm running in $P(n)$ time.
\end{definition}

\begin{definition}[The class \LP]
A distributed language
$\genericlang$ is in the class \LP if there exist a constant~$t$ and a polynomial $P$ such that $\genericlang$ can be decided by an \LD-algorithm with locality $t$ running in
$P(|N^{t}_{G,x,id}(v)|)$ time on each node $v$ of a graph $G$, where the IDs are $t$-locally unique.
\end{definition}

\begin{definition}[The class \LPstar]
A distributed language
$\genericlang$ is in the class \LPstar if there exist a constant~$t$ and a polynomial $P$ such that $\genericlang$ can be decided by an \LD-algorithm with locality $t$ running in
$P(|N^{t}_{G,x,id}(v)|)$ time on each node $v$ of a graph $G$, where the IDs are globally unique.
\end{definition}

With these definitions, we have $\LP \subseteq \LPstar.$

\subsubsection{Existential Quantifier: $\Sigma_1$-Type Classes}

\begin{definition}[The class $\NLD=\SigmaOneLocal$]
    A distributed language $\genericlang$ is
in the class $\NLD=\SigmaOneLocal$ if there exists an $\LD$-algorithm $\genericalgo$ such that for every configuration $(G, x)$:
\[(G, x) \in \genericlang \Rightarrow \exists c, \forall id, (\genericalgo\text{ accepts }(G, (x, c), id) ) ,\]
\[(G, x) \notin \genericlang \Rightarrow \forall c, \forall id, (\genericalgo\text{ rejects }(G, (x, c), id) ) .\]
\end{definition}

Note that the order of quantifier encodes what we discussed earlier: the identifiers are chosen adversarially after the certificates.

\begin{definition}[The class $\NPLD=\SigmaOnePLocal$]
    A distributed language $\genericlang$ is
in the class $\NPLD=\SigmaOnePLocal$ if there exist an $\LD$-algorithm $\genericalgo$ and polynomials $P$ and $Q$ such that for every $n$-node configuration~$(G, x)$:
\[(G, x) \in \genericlang \Rightarrow \exists c,  
\forall id, (\genericalgo\text{ accepts }(G, (x, c), id) ), \]
\[(G, x) \notin \genericlang \Rightarrow \forall c,  
\forall id, (\genericalgo\text{ rejects }(G, (x, c), id) ). \]
where in addition, \genericalgo must run in $P(n)$ time at each node, and $c$ must satisfy $|c(v)|\leq Q(n)$ for every node $v$.
\end{definition}
 
Note that the polynomials $P$ and $Q$ are chosen before the graph configuration, which is crucial in our proof of Lemma~\ref{lem:pip<pi}.
Also note that the nodes generally do not know $n$, and therefore cannot evaluate~$P(n)$, yet the algorithm $\genericalgo$ must be designed so that it runs in $P(n)$ time.

\begin{definition}[The class $\NLP=\SigmaOneLP$]
    A distributed language $\genericlang$ is
in the class $\NLP=\SigmaOneLP$ if there exist $t \in \mathbb{N}$, a $t$-round $\LD$-algorithm and polynomials $P$ and $Q$ such that for every $n$-node configuration~$(G, x)$ and every $t$-locally unique ID assignment:
\[(G, x) \in \genericlang \Rightarrow \exists c\leq Q, 
(\genericalgo\text{ accepts }(G, (x, c), id) ), \]
\[(G, x) \notin \genericlang \Rightarrow \forall c\leq Q, 
(\genericalgo\text{ rejects }(G, (x, c), id) ). \]
where in addition,
\genericalgo must run in $P (|N^{t}_{G,x,id}(v)|)$ time at each node $v$, 
and $\forall c\leq Q$ stands for $|c(v)|\leq Q (|N^{t}_{G,x,id}(v)|)$ for each~$v\in V$.
\end{definition} 

\subsubsection{Universal Quantifier: $\Pi_1$-Type Classes}

\begin{definition}[The class $\PiOneLocal$] 
A distributed language $\genericlang$ is
in the class $\PiOneLocal$ if there exists an \LD-algorithm $\genericalgo$ such that for every configuration $(G, x)$:
\[
(G, x) \in \genericlang \Rightarrow \forall c, \forall id, (\genericalgo\text{ accepts }(G, (x, c), id) ),\]
\[
(G, x) \notin \genericlang \Rightarrow \exists c, \forall id, (\genericalgo\text{ rejects }(G, (x, c), id) ).\]
\end{definition}

\begin{definition}[The class $\PiOnePLocal$] 
A distributed language $\genericlang$ is
in the class $\PiOnePLocal$ if there exists an \LD-algorithm $\genericalgo$ and polynomials $P$ and $Q$ such that for every $n$-node configuration~$(G, x)$: 
\[
(G, x) \in \genericlang \Rightarrow \forall c\leq Q, \forall id, (\genericalgo\text{ accepts }(G, (x, c), id) ),\]
\[
(G, x) \notin \genericlang \Rightarrow \exists c\leq Q, \forall id, (\genericalgo\text{ rejects }(G, (x, c), id) ).\]
where in addition, \genericalgo must run in $P(n)$ time at each node, and $\forall c\leq Q$ is a shorthand for $\forall c\in\{c:V\to\{0,1\}^\ast\mid\forall v\in V,\; |c(v)|\leq Q(n)\}$.
\end{definition}

\begin{definition}[The class $\PiOneLP$] 
A distributed language $\genericlang$ is
in the class $\PiOneLP$ if there exists $t \in \mathbb{N}$, a $t$-round \LD-algorithm $\genericalgo$ and polynomials $P$ and $Q$ such that for every $n$-node configuration~$(G, x)$ and every $t$-locally unique ID assignment: 
\[
(G, x) \in \genericlang \Rightarrow  \forall c\leq Q,(\genericalgo\text{ accepts }(G, (x, c), id) ),\]
\[
(G, x) \notin \genericlang \Rightarrow  \exists c\leq Q, (\genericalgo\text{ rejects }(G, (x, c), id) ).\]
where in addition, 
\genericalgo must run in $P (|N^{t}_{G,x,id}(v)|)$ time at each node $v$, 
and $\forall c\leq Q$ stands for 
$\forall c\in\{c:V\to\{0,1\}^\ast\mid\forall v\in V,\; |c(v)|\leq Q(|N^{t}_{G,x,id}(v)|)\}$.
\end{definition}

Note that there are two ways to define a bound on the certificate size after a universal quantifier: either we enumerate all the certificates \emph{of the given size}, or all the certificates \emph{up to a given size}. This is an important distinction since the first setting leaks more information about the size of the graph. We use the second definition.

In our definitions, we separately write the conditions for $(G,x) \in \genericlang$ and $(G,x) \notin \genericlang$ instead of using a single if-and-only-is statement, since we want the algorithm to either accept for all ID assignment, or to reject for all of them (see also~\cite[Sec.~1.1]{BalliuDFO18}).

\subsection{A Simple Observation}
\label{sec:lsp_defn}

Note that a consequence of the definitions is that $\PiOneLP,\SigmaOneLP \in \PH$ and $\PiOnePLocal,\SigmaOnePLocal \in \PH$ (remember that $\PH$ stands for the polynomial hierarchy in centralized complexity theory).
Now, we consider a language $\Lsp \in \textsf{EXPSPACE} \setminus \textsf{PSPACE},$ which exists by the space hierarchy theorem.
We construct a language $\node_{\Lsp}$ using $\Lsp$ as follows.
\[\node_{\Lsp}= \{(G,x) \mid G \text{ contains a single node labeled $x$, and } x \in \Lsp \} \]

It is easily seen that $\node_{\Lsp}\in \LD$ since a single node is allowed to use unbounded resources for its local computation.
On the other hand, $\node_{\Lsp} \notin \PH, $ since such an algorithm could be simulated by a centralized machine to compute $\Lsp$ in polynomial space (remember that $\PH \subseteq \textsf{PSPACE}$).
This yields the following corollary.
\begin{corollary}
\label{cor:PHvsLD}
    \[\textnormal{\node}_{\Lsp} \in\LD\setminus\big( \PiOneLP \cup \SigmaOneLP \cup \PiOnePLocal \cup \SigmaOnePLocal\big).\]
\end{corollary}
We notice that the corollary also holds if the nodes know $n$, 
since a centralized machine simulating the local algorithm knows the size of the graph.

\section{Polynomial in the Graph vs.\ in the Local View}
It can be easily seen that $\LPstar \subseteq \PLD.$ 
A node in an $\LPstar$ algorithm, by definition, is only allowed to use time polynomial in the size of its neighborhood. A node in a $\PLD$ algorithm, on the other hand, is permitted to use time polynomial in the size of the graph. The following result shows that these two notions overlap exactly. A $\PLD$ algorithm cannot see the whole graph and therefore cannot exploit the time permitted by its definition.
\label{sec:lp-pld}

    \begin{lemma}
    \label{lem:lpeqpld}
        $\LPstar = \PLD$.
    \end{lemma}
        
    \begin{proof}
    It is easy to see that $\LPstar \subseteq \PLD$:
    The size of a node's neighborhood is at most $n$, and the local input and ID of each node is polynomial in $n$.
    Therefore, the running time of an $\LPstar$ algorithm at a node is also polynomial in $n$,
    and the same algorithm is a $\PLD$ algorithm.

    To show $\PLD \subseteq \LPstar$ we consider a $\PLD$ algorithm $\genericalgo$ of radius $r$ on a graph $G$.
    Intuitively, we want to claim that the algorithm does not know the graph's size, and thus can only use time polynomial in its view.
    The formalization, however, is a bit trickier.

    Let $P$ be the polynomial bounding the running time of $\genericalgo$ as a function of the number $n$ of nodes. 
    Assume w.l.o.g.\ that $P$ is non-decreasing for all $n\geq 0$;
    if this is not the case, $P$ can be replaced by a bigger polynomial as follows.
    Let $P(n)=a_kn^k+\ldots+a_1n+a_0$, with $a_k\neq 0$.
    Clearly, $a_k>0$, as otherwise the algorithm's running time would have to be negative for a large enough $n$.
    Let $c$ be the largest local maxima of $P$ on any $n\geq0$, which must exist since $P$ is not monotone.
    Then $\bar P(n)=a_kn^k+c$ is monotone on the desired range, and satisfies $P(n)\leq \bar P(n)$ for all $n\geq0$, hence also bounds the local running time of $\genericalgo$.
    Thus, we can henceforth assume $P$ is monotonously increasing on $n\geq0$.
    
    We pick an arbitrary node $u \in V(G)$, and that the local computation time of $\genericalgo$ in $u$ is polynomial in the size of $u$'s $r$-neighborhood and input,
    hence $\genericalgo$ is also a legitimate $\LPstar$ algorithm with the same polynomial $P$ bounding its running time (but this time, as a function of the view size and not of~$n$).
    Let $\lambda_r$ be the total size of the labels in $u$'s $r$-neighborhood, including IDs.
    If the $r$-neighborhood of $u$ contains all of $G$ then \genericalgo's running time is bounded by $P(n)$, while an \LPstar algorithm is allowed to run for time $P(n+\lambda_r)$.
    As $P$ is monotone, we are done.
    
    Otherwise, 
    recall that $\#N^r[u]$ denotes the number of nodes in $u$'s $r$-neighborhood.
    If $\#N^r[u]+\lambda_r\geq n$ then \genericalgo's running time is still under $P(\#N^r[u]+\lambda_r)$, and we are done as before.
    
    Assume this is not the case.
    Let $v$ be a node in the $(r+1)$-neighborhood of $u$ but not in its $r$-neighborhood.
    Consider the graph $G'$ which is the same as $G$ on the $r$ neighborhood of $u$ and on $v$, and extended by a path of label-less nodes from $v$, so that the total number of nodes is 
    $n'=\#N^r[u]+\lambda_r$.
    Now, the view of $u$ in both $G$ and $G'$ is the same, and therefore the running time of $\genericalgo$ on both graphs must be identical, and bounded by 
    $P(\#N^r[u]+\lambda_r)=P(n')$.
    Thus, $\genericalgo$ is an \LPstar algorithm, as claimed.
    \end{proof}

Using lemma \ref{lem:lpeqpld}, we can give a simpler proof of the fact that $\PLD \subsetneqq \classP \cap \LD.$ This was proved in~\cite{TshuvaO23} using more involved, ad hoc arguments.
\begin{lemma}\label{lem:PLDsneLD}
    $\PLD \subsetneqq \classP \cap \LD$.
\end{lemma}
By Lemma~\ref{lem:lpeqpld}, it suffices to prove that 
$\LPstar \subsetneqq \classP \cap \LD$, which we do in the following two claims.
        \begin{claim}
            $\LPstar \subseteq \classP \cap \LD$.
        \end{claim}
        \begin{proof}
            Every $\LPstar$ algorithm can be simulated by a centralized machine in $\poly  n$ time. Additionally, every $\LPstar$ algorithm is also an $\LD$ algorithm.
        \end{proof}
        \begin{claim}
            $\LPstar \neq \classP \cap \LD$.
        \end{claim}

        \begin{proof}
           To show the separation, we consider a language $\genericlang\in \textsf{EXP} \setminus \classP$, which exists by the time hierarchy theorem.
            Using $\genericlang$, we define
            \begin{multline*}
            \pathlang_{\genericlang} = \{ G \mid G \text{ is an $n$-node path with labels }(n,n),\ldots, (n,2), (n,1,\varphi)\\ \text{ where $\varphi$ is a word of size at most $\log{n}$ in $\genericlang$}\}.
            \end{multline*}
            We assume, for the sake of contradiction, that there exists an $\LPstar$ algorithm $\genericalgo$ of locality radius $r$ for $\pathlang_\genericlang$.

            We consider the following centralized algorithm for determining if a word $\varphi$ satisfies $\varphi \in \genericlang$. 
            We consider all $n$ large enough such that $\log n > 2r$ and such that $\psi,\psi'$ used later in the proof are of size at most $\log n$ each.

            \begin{itemize}
                \item Consider the path graph $G'$ with labels $(n,\log n),(n,\log n -1),\cdots,(n,1,\varphi)$
                \item Simulate $\genericalgo$ on $G'$
                \item Accept if and only if all the node labeled $(n,r),(n,r-1),\cdots,(n,1,\varphi)$ accept. 
                We call these nodes $u_r,u_{r-1},\cdots,u_{1}$.
            \end{itemize}
            To show the correctness of our algorithm we notice that, in the case where $\genericalgo$ accepts $G$ 
            , all nodes in $G$ accept and therefore $u_r,\cdots,u_1$ must also accept in our algorithm in $G'$ since their local views are the same.

            Now we consider the case that $\genericalgo$ rejects $G$. We can assume that there exists two words $\psi,\psi'$ such that $\psi \in \genericlang$ and $\psi' \notin \genericlang$: both $\varnothing$ and the set $\{0,1\}^*$ are in $\classP$, while $\genericlang\notin\classP$, so $\genericlang$ is neither of them.

            Consider the two paths with labels
            \[G_1: \: (n,n), \ldots, (n,r), \ldots, (n,1,\psi)\,\]
            \[G_2: \: (n,n), \ldots, (n,r), \ldots, (n,1,\psi')\]

            The local view of all nodes $(n,i)$ with $i > r$ is the same in both instances, but $G_1$ should be accepted and while $G_2$ should be rejected. 
            Therefore, all of these nodes must accept (in both cases), and the only rejecting nodes can be the ones with $i \le r$. Therefore, our algorithm sees a rejecting node in one of $u_r,\cdots,u_1$.

            Since $\genericalgo$ is an $\LPstar$ algorithm, it  runs in each node in time  polynomial with respect to the size of its neighborhood, which is of size $O(\log n)$.
            Therefore, our algorithm simulates $2r$ many executions of $\genericalgo$, each taking $\text{poly}(\log n)$ time, for a total of $\text{poly}(\log n)$ time. 
            Hence, our centralized algorithm $\genericlang'$ is a poly-time verifier for $\genericlang$: it runs in $\text{poly}(\log n)$ time on inputs of size $\log n$.
            But we assumed $\genericlang \in \textsf{EXP}\setminus\classP$, a contradiction. Therefore, such a \PLD algorithm $\genericalgo$ cannot exist.
        \end{proof}

\section{Universal Certificates and the Graph Size}
\label{sec:poly-certificates}

In this section we discuss universal certificates. 
We show that giving certificates of length polynomial in the size of the graph end up ``leaking'' information to the nodes. 
We then show that this only gives a lower bound on $n$, knowing $n$ exactly makes the model even more powerful.
Unlike using certificates that are polynomial in $n$, we show that using certificates polynomial in the view of each node behaves ``as expected'', i.e., is strictly weaker than having unbounded certificates.
We conclude with a few implications of the aforementioned results, at the form of extremely  simple proofs of several separation results.

\subsection{Bounded Certificates Leak Information}    
\label{subsec:bounded-leak}

One might guess that having a bound on the size of certificates limits the power of the model, and this is indeed true in many cases. 
The next result, however, shows the opposite may also happen: the size limit might enable to solve problems which couldn't be solved with unbounded certificates.
    It is easy to see that $\node_{\Lsp} \in \PiOneLocal \setminus \PiOnePLocal.$
    The next result shows that $\PiOnePLocal$ and $\PiOneLocal$ are  incomparable.

\begin{lemma}\label{lem:pip<pi}
        $\PiOnePLocal \not\subset \PiOneLocal$.
    \end{lemma}
    We consider the language $\AGTG$ (All Greater Than size of Graph) of $n$-node path with labels
    \[x_1,\cdots,x_n\]
    such that for all $i$, $x_i > n$. 
    We show that this language is in $\PiOnePLocal$ but not in $\PiOneLocal$.

    \begin{claim}
        $\AGTG\in \PiOnePLocal$.
    \end{claim}
    \begin{proof}
        Recall that a language is in $\PiOnePLocal$ if there exists a polynomial $Q$ and an algorithm that uses certificates up to size $Q(n)$ on an $n$-node graph.
        To prove $\AGTG\in \PiOnePLocal$, we use $Q(n)=n$ as the polynomial bounding the certificates, and the following local decision algorithm.
        
        For a node $u \in V$, let $C(u)$ denote the certificate it receives, and $|C(u)|$ its length.
        Node $u$ accepts if $x_u>|C(u)|$, and rejects otherwise. 
        If $x_u>n$ then $x_u>|C(u)|$ for all the certificates, and $u$ always accepts; if $x_u>n$ for all nodes $u$, then they all accept.
        Conversely, if $x_v\leq n$ for some node $v$ then this node will reject, e.g., for a certificate with $|C(v)|=n$.
    \end{proof} 

    \begin{claim}
        $\AGTG\notin \PiOneLocal$.
    \end{claim}
    \begin{proof}
    Assume for contradiction that there is an algorithm \genericalgo for $\AGTG$, with locality $r$.
    Consider three path graphs with input labels as follows, where we take $n$ large enough so that $n > 2r$.
    \begin{itemize}
        \item $G_1\in \AGTG$ is an $n$-node graph with labels
        \[n+1,n+2,\ldots,n+2\hspace{6ex}\]
        \item $G_2\in \AGTG$ is an $(n+1)$-node graph with labels
        \[ n+2,n+2,\ldots,n+2,n+2\]
        \item $G_3\notin \AGTG$ is an $(n+1)$-node graph with labels
        \[ n+1,n+2,\ldots,n+2,n+2\]
    \end{itemize}
    In all these graphs, the IDs start by 1 on the left and grow consecutively.

    Note that $G_1,G_2\in \AGTG$, so in both graphs, any assignment of certificates to the nodes make all of them accept.
    Consider an assignment of certificates to $G_3$.
    The nodes $1,\ldots,r+1$ have the same view as in the same assignment to $G_1$, and thus they all accept.
    The other nodes of $G_3$, $r+2,\ldots,n+1$ have the same view as in the same assignment to $G_2$, so they all accept as well.
    Hence, all nodes must accept under all certificate assignments;
    however, $G_3\notin \AGTG$, rendering $\genericalgo$ wrong.
    \end{proof}
    
    While $\AGTG\in\PiOnePLocal$ was a surprising result that required some work, 
It can be easily seen that  $\AGTG\in\SigmaOnePLocal$ (i.e., \NPLD).

\begin{claim}
        $\AGTG \in \NPLD$.
    \end{claim}    
    \begin{proof}
        The $i^{\text{th}}$ node is given the certificate $(1^i,1^n)$. The nodes can then check that the certificates are increasing consistently and consecutively among their neighbors. 
        A node with only one neighbor then checks that its certificate is $(1,1^n)$ and its neighbor's certificate is $(11,1^n)$, or its certificate is $(1^n,1^n)$ and its neighbors certificate is $(1^{n-1},1^n)$.
        After verifying this, all nodes have the correct value of $n$ and can make sure that their inputs are greater than $n$.
    \end{proof}

With this claim in hand, 
the language \AGTG can also be used to give a very simple proof of the separation $\PLD \subsetneqq \PiOnePLocal \cap \NPLD$, which was also claimed in \cite{TshuvaO23}.

\begin{corollary}
\label{cor:PLD in intersection}
    $\PLD \subsetneqq \PiOnePLocal \cap \NPLD$.
\end{corollary}
\begin{proof}
    The containment is immediate from the definitions. 
    We already know from Lemma~\ref{lem:pip<pi} that $\AGTG \in \PiOnePLocal \setminus \PiOneLocal$ and we know that $\PLD \subseteq \LD \subseteq \PiOneLocal$.
    As $\AGTG \in \NPLD$, the claim follows.    
\end{proof}
  In Appendix~\ref{app:proof_of_pld_subsetneq} we give another proof of the aforementioned result. 
  While our proof above uses locality, the proof in appendix is based on ideas from \cite{BalliuDFO18} and is reliant on computational constraints. 

\subsection{Knowing $n$ Exactly}
\label{subsec:knowing n}

    Recall that $\Pi_1^{\textsf{P-local}[n]}$ is defined similarly to $\PiOnePLocal$, but giving each node the value of $n$ as additional information.    
    We can easily see that
    $\PiOnePLocal \subseteq \Pi_1^{\textsf{P-local}[n]}$:
    given a $\PiOnePLocal$ algorithm \genericalgo, we can also consider \genericalgo as a $\Pi_1^{\textsf{P-local}[n]}$  algorithm that ignores $n$.
    The previous result implies that polynomial certificates leak some information about $n$, so one may suspect that $\PiOnePLocal = \Pi_1^{\textsf{P-local}[n]}$. 
    The following lemma proves this suspicion to be wrong.
\begin{lemma}
\label{lem:piloc<pilocn}
    $\PiOnePLocal \subsetneqq \Pi_1^{\Plocal[n]}.$
\end{lemma}
We consider the language $\ALTG$ (All Lesser Than size of Graph) of $n$-node paths with labels
        \[x_n,x_{n-1},\ldots,x_1\]
        such that for all $i$, $x_i < n$. 
        We show that this language is in $\Pi_1^{\textsf{P-local}[n]}$ but not in $\PiOnePLocal$.
    \begin{claim}
        $\ALTG \in \Pi_1^{\Plocal[n]}$.
    \end{claim}
    \begin{proof}
        Each node can check the condition directly since it knows the size of the graph.
    \end{proof}
    \begin{claim}
        $\ALTG \notin \PiOnePLocal$.
    \end{claim}
    \begin{proof}
        Assume there exists a $\PiOnePLocal$ algorithm \genericalgo for $\ALTG$ with locality radius $r$. 
        Let $Q(n)$ be the bound on certificate sizes as a function of $n$.

        Take $n$ large enough so that $n>2r$ and $Q(n+1) \geq Q(n)$, and consider the following instances
        \begin{itemize}
            \item $G_1\in \ALTG$ is an $(n+1)$-node graph labeled
            \[n,\hspace{3ex}n-1,\ldots,n-1,n-1;\]
            \item $G_2\in \ALTG$ is an $n$-node graph labeled \[n-1,n-1,\ldots,n-1;\hspace{5.5ex}\]
            \item $G_3\notin \ALTG$ is an $n$-node graph labeled
            \[n,\hspace{3ex}n-1,\ldots,n-1.\hspace{6ex}\]
        \end{itemize}
          We assume the IDs of the nodes start with $1$ on the left and increase consecutively.
          $G_1$~and~$G_2$ are clearly accepting instances,
          so in both, all nodes accept under all certificate assignments.
          
          Now, the view of the $r+1$ left-most vertices is the same in both $G_1$ and $G_3$, so they must behave the same in both graphs. Similarly, the views of the nodes from $r+1$ to $n$ is the same in both $G_2$ and $G_3$, so they must behave the same in those graphs. 
          Therefore, for any certificate assignment, each node of $G_3$ must behave as a node of $G_1$ or of $G_2$ under the same certificate assignment (and this assignment can occur in $G_1$ since $Q(n+1) \geq Q(n)$), and accept.
          However, 
          $G_3\notin \ALTG$,
          a contradiction.
      \end{proof}

      We now show, as corollary of this lemma and Lemma~\ref{lem:pip<pi}, that $ \Pi_1^{\Plocal[n]} $ and $ \PiOneLocal$ are incomparable.

      \begin{corollary}
      \label{cor:pioneplocn<pioneloc}
      $ \Pi_1^{\Plocal[n]} $ and $ \PiOneLocal$ are incomparable.      
      \end{corollary}
      \begin{proof}
      We notice that $\node_{\Lsp} \in \PiOneLocal \setminus \Pi_1^{\Plocal[n]}$ from Corollary~\ref{cor:PHvsLD}, implying $ \PiOneLocal \not\subset \Pi_1^{\Plocal[n]}$.
      On the other hand,
        if $ \Pi_1^{\Plocal[n]} \subseteq \PiOneLocal$
        then with Lemma~\ref{lem:piloc<pilocn} above, we get $ \Pi_1^{\Plocal} \subseteq \PiOneLocal$.
        But we have proved this to be false in Lemma~\ref{lem:pip<pi}, so $ \Pi_1^{\Plocal[n]} $ and $ \PiOneLocal$ are incomparable.
      \end{proof}

\subsection{Certificates Polynomial in the View} 
\label{subsec:LPcertificates}

In Lemma~\ref{lem:pip<pi} we show that $\PiOnePLocal$ is not contained in $\PiOneLocal$.
We now examine the relation of $\PiOneLP$ to both these classes.
The idea for the inclusions $\PiOneLP \subseteq \PiOneLocal$ and $\PiOneLP \subseteq \PiOnePLocal$ was outlined in ~\cite{Reiter24}. We prove the inclusions here for completeness and also additionally show that the inclusions are strict.

We start by proving 
$\PiOneLP \subsetneqq \PiOneLocal$.
In fact, we prove something stronger: 
$\PiOneLP \subsetneqq \LD.$\begin{lemma}\label{lem:pi<ld}
        $\PiOneLP \subsetneqq \LD \subseteq \PiOneLocal$.
    \end{lemma}
    \begin{proof}
    We start by proving $\PiOneLP \subseteq \LD$.
        Consider a $\Pi_1^{\textsf{LP}}$ algorithm \genericalgo with locality radius $r$ and where the bound on the certificates is $Q(x)$ as a function of the size of the view of each node. 
        We construct an $\LD$ algorithm $\genericalgo'$ of locality radius $2r$ as follows.
        
        Each node $v \in V(G)$ checks its $2r$-neighborhood $N^{2r}(v)$.
        For each node $u$ in its $r$-neighborhood,
        $v$ can compute the size of $u$'s distance-$r$ neighborhood, $|N^{r}(u)|$.
        Then, $v$ simulates \genericalgo on all assignments of certificates to itself and its distance-$r$ neighborhood, where node $u \in N^{r}(v)$ can get certificates of size at most $Q(|N^{r}(u)|)$ .
        Node $v$  accepts if and only if the simulation of \genericalgo accepts on all these certificate.
        
        Our new algorithm is clearly in $\LD$, since there are only finitely many possible certificate.
        To show the correctness of this algorithm, we notice that if 
        \genericalgo is executed on an accepting instance, then every node must accept, and do so for all possible assignments of certificates in its neighborhood;
        hence, every node will accept in $\genericalgo'$ as well. 
        On the other hand, in case of a rejecting instance, there exists a node $v$ and an assignment of certificates to its neighborhood, such that $v$ rejects in $\genericalgo$, and consequently, in $\genericalgo'$ as well. 

        To show that the inclusion is strict, we consider $\node_{\Lsp}.$
        It is easily seen that $\node_{\Lsp}\in \LD$ since a single node is allowed to use unbounded resources for its local computation. 
        On the other hand, $\node_\genericlang \notin \PiOneLP$.
    \end{proof}

Next, we show similar results (containment and strictness) for 
$\PiOnePLocal$.

\begin{lemma}
\label{lem:pilp<pioneploc}
    $\PiOneLP \subsetneqq \PiOnePLocal$.
\end{lemma}
 
\begin{proof}
   Given a $\PiOneLP$ algorithm \genericalgo with locality radius $r$ and the bound on the certificates being $Q(x)$ as a function of the size of the view of each node.
    Similar to Lemma~\ref{lem:lpeqpld}, we assume that $Q'$ is a non-decreasing polynomial such that $Q'(x) \ge Q(x)$ for all $x \ge 0$.
    We construct a $\PiOnePLocal$ algorithm $\genericalgo'$ of radius $2r$ and certificates are bounded by $R(x) = Q'(x^2h(x))$ where $h(x)$ is a polynomial upper bound on the size of labels.
    
     Each node $v \in V(G)$ aggregates its $2r$-neighborhood $N^{2r}(v)$.
    For each node $u$ in its $r$-neighborhood,
    $v$ can compute the size of $u$'s distance-$r$ view, denoted $|N^r_{G,x,{id}}(u)|$.

    Node $v$ first examines the certificate sizes: if for some $u \in N^r(v)$ (including $v$ itself), $|c(u)| > Q(|N^r_{G,x,{id}}(u)|)$ then $v$ accepts. 
    Otherwise $v$ runs \genericalgo on its $r$-neighborhood and outputs accordingly.

    Clearly, $\genericalgo'$ is a legitimate $\PiOnePLocal$ algorithm.
    To show the correctness of $\genericalgo'$, we notice that if \genericalgo is executed in an accepting instance then all nodes must accept. 
    For $\genericalgo'$ on the same instance and $v \in V(G)$, either one of the certificates in its $r$-neighborhood was too large (in which case it accepts) or \genericalgo is simulated and $\genericalgo'$ also accepts.
    On the other hand, for a rejecting instance, there must be a node $v$ and an assignment of certificates to its neighborhood such that $v$ rejects when simulating \genericalgo.
    We have taken $Q'$ to be non-decreasing $Q'(x) \ge Q(x)$ for all $x \ge 0$ and $n^2h(n) \ge |N^r_{G,x,{id}}(u)|$ for all $u$. Therefore while simulating $\genericalgo$ all the certificates with size $|c(u)| \le  Q(|N^r_{G,x,{id}}(u)|)$ will be considered. Therefore, the rejecting certificates must also be accounted for when running $\genericalgo'$ and therefore must result in rejection.
    Unlike Lemma~\ref{lem:lpeqpld}, we couldn't directly assume $Q$ is a non-decreasing function since $\genericalgo$ is an arbitrary algorithm and we can only ensure it behaves as expected if certificates are bounded by $Q.$

    For the separation,
    recall that in the proof of Lemma~\ref{lem:pip<pi} we have shown a language $\AGTG$ satisfying $\AGTG\in \PiOnePLocal$ but $\AGTG\notin \PiOneLocal$. Lemma~\ref{lem:pi<ld} asserts that $\PiOneLP \subseteq \PiOneLocal$, and hence $\AGTG\notin \PiOneLP$. The claim follows.
    \end{proof}

\bibliographystyle{alpha}
\bibliography{local-decision-refs}

\appendix

\section{An Alternative Proof of Claim~\ref{cor:PLD in intersection}}
\label{app:proof_of_pld_subsetneq}

Recall that Claim~\ref{cor:PLD in intersection} asserts that $\PLD \subsetneqq \PiOnePLocal \cap \NPLD$.
We give another proof of this claim.
The containment is immediate, so we only have to show separation, which we do next, using the language $\texttt{ITER}$ defined in \cite{BalliuDFO18}.

\begin{claim}
$\texttt{ITER} \in (\NPLD \cap \PiOnePLocal)\setminus \LD$.
\end{claim}

\begin{proof}
We follow the footsteps of in~\cite{BalliuDFO18}, who defined a language $\texttt{ITER}$ to separate $\Pi_1$ from \LD.
The language it composed of labeled paths of the following structure.
One node $p$ is the \emph{pivot} node, and is labeled by a Turing machine $M$ and two inputs $a$ and $b$ to it. 
Each other node $u$ is labeled with the machine $M$, with  its distance $d(u,p)$ from the pivot $p$, and with the configuration of $M$ after $d(u,p)$ steps: when starting from $a$ if it is on the left of $p$ (denoted $\text{config}(M,a,d(u,p))$), and when starting from $b$ when it is on its right (denoted $\text{config}(M,b,d(u,p))$).
Here, the sides are arbitrary, and need only be consistent among the different nodes.
Finally, the endpoints of the path must both contain halting configurations, and at least one of these configurations must also be accepting.

We define a new language $\texttt{ITER}^-$ which is the same as $\texttt{ITER}$ but with the condition that both states at the endpoints only need to be halting.
To show that $\texttt{ITER}^- \in \PiOnePLocal$, we present the following algorithm. 
Each node first checks if all the nodes got the same machine and if the distances they received is consistent with each other. 
The pivot checks if one of its neighbors got $a$ and the other got $b$ as the input to $M$. The nodes then check if the Turing machine configurations are consecutive according to their distance from the pivot. 
The nodes with degree 1 check if their configuration is halting. 
A node rejects if any of these checks fails.

To check $\texttt{ITER}$, we use a similar algorithm, with certificate lengths bounded by the identity polynomial $Q(n)=n$.
Each node first checks if $(G,x) \in \texttt{ITER}^-$, and each node but the pivot decides according to this check.
If the pivot would have rejected in the check for $\texttt{ITER}^-$, it rejects for $\texttt{ITER}$ as well.
Otherwise, the pivot interprets the length of the certificate as a non-negative integer $k \le n$, where $n$ is the number of nodes.
The pivot rejects if both $\text{config}(M,b,k)$ and $\text{config}(M,a,k)$ are  rejecting configurations, and accepts otherwise.
For correctness, first note that the check for $\texttt{ITER}^-$ ensures the syntactic correctness, and that both endpoints of the path have halting configurations. 
If both are rejecting, then the machine halts and rejects on both $a$ and $b$ in less than $n$ steps (for each), as the path length is $n$.
In this case, when the pivot gets a certificate of some length $k\le n$, it will see that both $\text{config}(M,b,k)$ and $\text{config}(M,a,k)$ are rejecting, and will reject.
Finally, we notice that the algorithm runs in polynomial local time (polynomial in $n$).
Therefore, $\texttt{ITER} \in \PiOnePLocal$. 

Proving $\texttt{ITER} \in \NPLD$ directly is not hard. However, we take a different path, utilizing known results regarding relationships between classes of tasks.
First, we notice that it is easy to check in a centralized manner that a labeled graph is in $\texttt{ITER}^-$, by checking the syntax of the labels and simulating $M$ on $a$ and $b$.
Additionally, it is trivial to check that the states at both endpoints are halting and that at least one of them is accepting. 
All this can be done deterministically in time polynomial in the length of the path and the encoding of the machine and states, and therefore $\texttt{ITER} \in \classP \subseteq \classNP$.
We have shown $\texttt{ITER} \in \PiOnePLocal$ and we know that $\PiOnePLocal \subseteq \NLD$, so
$\texttt{ITER}\in \NLD$.
We also know that 
and $\NPLD = \NLD \cap \classNP$,
and therefore we conclude $\texttt{ITER} \in \NPLD$.

Finally, it was proved in in~\cite[Proposition~7]{BalliuDFO18} that $\texttt{ITER} \notin \LD$, and this holds in our case as well.
\end{proof}

\end{document}